\documentclass[12pt,draftclsnofoot,journal,onecolumn]{IEEEtran}

\usepackage{cite}
\usepackage{graphicx}
\usepackage{amsmath}
\usepackage{amsfonts}
\usepackage{array}
\usepackage{amssymb}
\usepackage{subfigure}
\usepackage{booktabs}
\usepackage[alwaysadjust]{paralist}
\newtheorem{lemma}{Lemma}
\newtheorem{theorem}{Theorem}
\newtheorem{corollary}{Corollary}
\newtheorem{remark}{Remark}

\begin{document}

\title{Mode Switching for MIMO Broadcast Channel Based on Delay and Channel Quantization}
\author{Jun Zhang,  Robert W. Heath Jr., Marios Kountouris, and Jeffrey G. Andrews
\thanks{The authors are with the Wireless Networking and Communications
Group, Department of Electrical and Computer Engineering, The
University of Texas at Austin, 1 University Station C0803, Austin,
TX 78712--0240. Email: \{jzhang2, rheath, mkountouris, jandrews\}@ece.utexas.edu. This work has been supported in part by AT\&T Labs, Inc.}}

\maketitle

\begin{abstract}
Imperfect channel state information degrades the performance of
multiple-input multiple-output (MIMO) communications; its effect on
single-user (SU) and multi-user (MU) MIMO transmissions are quite
different. In particular, MU-MIMO suffers from residual inter-user
interference due to imperfect channel state information while
SU-MIMO only suffers from a power loss. This paper compares the
throughput loss of both SU and MU MIMO on the downlink due to delay
and channel quantization. Accurate closed-form approximations are
derived for the achievable rates for both SU and MU MIMO. It is
shown that SU-MIMO is relatively robust to delayed and quantized
channel information, while MU-MIMO with zero-forcing precoding loses
spatial multiplexing gain with a fixed delay or fixed codebook size.
Based on derived achievable rates, a mode switching algorithm is
proposed that switches between SU and MU MIMO modes to improve the
spectral efficiency, based on the average signal-to-noise ratio
(SNR), the normalized Doppler frequency, and the channel
quantization codebook size. The operating regions for SU and MU
modes with different delays and codebook sizes are determined, which
can be used to select the preferred mode. It is shown that the MU
mode is active only when the normalized Doppler frequency is very
small and the codebook size is large.
\end{abstract}

\begin{keywords}
Multi-user MIMO, adaptive transmission, mode switching, imperfect
channel state information at the transmitter (CSIT), zero-forcing
precoding.
\end{keywords}

\section{Introduction}
Over the last decade, the point-to-point multiple-input
multiple-output (MIMO) link (SU-MIMO) has been extensively
researched and has transited from a theoretical concept to a
practical technique \cite{Telatar99,Goldsmith03}. Due to space and
complexity constraints, however, current mobile terminals only have
one or two antennas, which limits the performance of the SU-MIMO
link. Multi-user MIMO (MU-MIMO) provides the opportunity to overcome
such a limitation by communicating with multiple mobiles
simultaneously. It effectively increases the number of equivalent
spatial channels and provides spatial multiplexing gain proportional
to the number of transmit antennas at the base station even with
single-antenna mobiles. In addition, MU-MIMO has higher immunity to
propagation limitations faced by SU-MIMO, such as channel rank loss
and antenna correlation \cite{GesKou07SPMag}.

There are many technical challenges that must be overcome to exploit
the full benefits of MU-MIMO. A major one is the requirement of
channel state information at the transmitter (CSIT), which is
difficult to get especially for the downlink/broadcast channel. For
the MIMO downlink with $N_t$ transmit antennas and $N_r$ receive
antennas, with full CSIT the sum throughput can grow linearly with
$N_t$ even when $N_r=1$, but without CSIT the spatial multiplexing
gain is the same as for SU-MIMO, i.e. the throughput grows linearly
with $\min(N_t,N_r)$ at high SNR \cite{HasSha07JSAC}. Limited
feedback is an efficient way to provide partial CSIT, which feeds
back the quantized channel information to the transmitter via a
low-rate feedback channel \cite{LovHea04CommMag,LovHea08JSAC}.
However, such imperfect CSIT will greatly degrade the throughput
gain provided by MU-MIMO \cite{Jin06IT,DinLov07Tsp}. Besides
quantization, there are other imperfections in the available CSIT,
such as estimation error and feedback delay. With imperfect CSIT, it
is not clear whether--or more to the point, when-- MU-MIMO can
outperform SU-MIMO. In this paper, we compare SU and MU-MIMO
transmissions in the MIMO downlink with CSI delay and channel
quantization, and propose to switch between SU and MU MIMO modes
based on the achievable rate of each technique with practical
receiver assumptions.

\subsection{Related Work}
For the MIMO downlink, CSIT is required to separate the spatial
channels for different users. To obtain the full spatial
multiplexing gain for the MU-MIMO system employing zero-forcing (ZF)
or block-diagonalization (BD) precoding, it was shown in
\cite{Jin06IT,RavJin07JSAC} that the quantization codebook size for
limited feedback needs to increase linearly with SNR (in dB) and the
number of transmit antennas. Zero-forcing dirty-paper coding and
channel inversion systems with limited feedback were investigated in
\cite{DinLov07Tsp}, where a sum rate ceiling due to a fixed codebook
size was derived for both schemes. In \cite{YooJin07JSAC}, it was
shown that to exploit multiuser diversity for ZF, both channel
direction and information about signal-to-interference-plus-noise
ratio (SINR) must be fed back. More recently, a comprehensive study
of the MIMO downlink with ZF precoding was done in
\cite{CaiJin07Submit}, which considered downlink training and
explicit channel feedback and concluded that significant downlink
throughput is achievable with efficient CSI feedback. For a compound
MIMO broadcast channel, the information theoretic analysis in
\cite{CaiJin07Asilomar} showed that scaling the CSIT quality such
that the CSIT error is dominated by the inverse of the SNR is both
necessary and sufficient to achieve the full spatial multiplexing
gain.

Although previous studies show that the spatial multiplexing gain of
MU-MIMO can be achieved with limited feedback, it requires the
codebook size to increase with SNR and the number of transmit
antennas. Even if such a requirement is satisfied, there is an
inevitable rate loss due to quantization error, plus other CSIT
imperfections such as estimation error and delay. In addition, most
of prior work focused on the achievable spatial multiplexing gain,
mainly based on the analysis of the rate loss due to imperfect CSIT,
which is usually a loose bound
\cite{Jin06IT,RavJin07JSAC,CaiJin07Asilomar}. Such analysis cannot
accurately characterize the throughput loss, and no comparison with
SU-MIMO has been made. In this paper, we derive good approximations
for the achievable throughput for both SU and MU MIMO systems with
fixed channel information accuracy, i.e. with a fixed delay and a
fixed quantization codebook size. We are interested in the following
question: \emph{With imperfect CSIT, including delay and channel
quantization, when can MU-MIMO actually deliver a throughput gain
over SU-MIMO?} Based on this, we can select the one with the higher
throughput as the transmission technique.

\subsection{Contributions}
In this paper, we investigate SU and MU-MIMO in the broadcast
channel with CSI delay and limited feedback. The main contributions
of this paper are as follows.
\begin{itemize}
\item \textbf{SU vs. MU Analysis}. We investigate the impact of imperfect CSIT due to delay and channel
quantization. We show that the SU mode is more robust to imperfect
CSIT as it only suffers a constant rate loss, while MU-MIMO suffers
more severely from the residual inter-user interference. We
characterize the residual interference due to delay and channel
quantization, which shows these two effects are equivalent. Based on
an independence approximation of the interference terms and the
signal term, accurate closed-form approximations are derived for the
ergodic rates for both SU and MU MIMO modes.

\item \textbf{Mode Switching Algorithm.} A SU/MU
mode switching algorithm is proposed based on the ergodic sum rate
as a function of the average SNR, normalized Doppler frequency, and
the quantization codebook size. This transmission technique only
requires a small number of users to feed back instantaneous channel
information. The mode switching points can be calculated from the
previously derived approximations for ergodic rates.

\item \textbf{Operating Regions.} The \emph{operating regions} for SU and
MU modes are determined, from which we can determine the active mode
and find the condition that activates each mode. With a fixed delay
and codebook size, if the MU mode is possible at all, there are two
mode switching points, with the SU mode preferred at both low and
high SNRs. The MU mode will only be activated when the normalized
Doppler frequency is very small and the codebook size is large. From
the numerical results, the minimum feedback bits per user to get the
MU mode activated grow approximately linearly with the number of
transmit antennas.

\end{itemize}

The rest of the paper is organized as follows. The system model and
some assumptions are presented in Section \ref{Sec:SysMod}. The
transmission techniques for both SU and MU MIMO modes are described
in Section \ref{Sec:TransTech}. The rate analysis for both SU and MU
modes and the mode switching are done in Section \ref{Sec:ModSel}.
Numerical results and conclusions are in Section \ref{Sec:Num} and
\ref{Sec:Con}, respectively.

\section{System Model}\label{Sec:SysMod}
We consider a MIMO downlink, where the transmitter (the base
station) has $N_t$ antennas and each mobile user has a single
antenna. The system parameters are listed in Table \ref{Tbl1}.
During each transmission period, which is less than the channel
coherence time and the channel is assumed to be constant, the base
station transmits to one (SU-MIMO mode) or multiple (MU-MIMO mode)
users. The discrete-time complex baseband received signal at the
$u$-th user at time $n$ is given as\footnote{In this paper, we use
uppercase boldface letters for matrices ($\mathbf{X}$) and lowercase
boldface for vectors ($\mathbf{x}$). $\mathbb{E}[\cdot]$ is the
expectation operator. The conjugate transpose of a matrix
$\mathbf{X}$ (vector $\mathbf{x}$) is $\mathbf{X}^*$
($\mathbf{x}^*$). Similarly, $\mathbf{X}^\dag$ denotes the
pseudo-inverse, $\tilde{\mathbf{x}}$ denotes the normalized vector
of $\mathbf{x}$, i.e.
$\tilde{\mathbf{x}}=\frac{\mathbf{x}}{\|\mathbf{x}\|}$, and
$\hat{\mathbf{x}}$ denotes the quantized vector of
$\tilde{\mathbf{x}}$.}
\begin{equation}
{y}_u[n]=\mathbf{h}^*_u[n]\sum_{u'=1}^{U}\mathbf{f}_{u'}[n]{x}_{u'}[n]+{z}_u[n],
\end{equation}
where $\mathbf{h}_u[n]$ is the $N_t\times{1}$ channel vector from
the transmitter to the $u$-th user, and $z_u[n]$ is the normalized
complex Gaussian noise vector, i.e. $z_u[n]\sim\mathcal{CN}(0,1)$.
${x}_u[n]$ and $\mathbf{f}_u[n]$ are the transmit signal and
$N_t\times{1}$ precoding vector for the $u$-th user, respectively.
The transmit power constraint is
$\mathbb{E}\left\{\mathbf{x}^*[n]\mathbf{x}[n]\right\}=P$, where
$\mathbf{x}[n]=[x^*_1,x^*_2,\cdots,x^*_U]^*$. As the noise is
normalized, $P$ is also the average transmit SNR.

To assist the analysis, we assume that the channel $\mathbf{h}_u[n]$
is well modeled as a spatially white Gaussian channel, with entries
${h}_{i,j}[n]\sim\mathcal{CN}(0,1)$, and the channels are i.i.d.
over different users. The results will be different for different
channel models. For example, a limited feedback system with line of
sight MIMO channel requires fewer feedback bits compared to the
Rayleigh channel \cite{RavJin07Glob}. The investigation of other
channel models is left to future work.

We consider two of the main sources of the CSIT imperfection--delay
and quantization error\footnote{For a practical system, the feedback
bits for each user is usually fixed, and there will inevitably be
delay in the available CSI, both of which are difficult or even
impossible to adjust. Other effects such as channel estimation error
can be made small such as by increasing the transmit power or the
number of pilot symbols.}, specified as follows.

\subsection{CSI Delay Model} We consider a stationary ergodic Gauss-Markov block fading process \cite[Sec. 16--1]{Hay96},
where the channel stays constant for a symbol duration and changes
from symbol to symbol according to
\begin{equation}\label{ChDelay}
\mathbf{h}[n]=\rho{\mathbf{h}}[n-1]+\mathbf{e}[n],
\end{equation}
where $\mathbf{e}[n]$ is the channel error vector, with i.i.d.
entries $e_{i}[n]\sim\mathcal{CN}(0,\epsilon_e^2)$, and it is
uncorrelated with $\mathbf{h}[n-1]$. We assume the CSI delay is of
one symbol. It is straightforward to extend the results to the
scenario with a delay of multiple symbols. For the numerical
analysis, the classical Clarke's isotropic scattering model will be
used as an example, for which the correlation coefficient is
$\rho=J_0(2\pi{f_d}T_s)$ with Doppler spread $f_d$ \cite{Cla68},
where $J_0(\cdot)$ is the zero-th order Bessel function of the first
kind. The variance of the error vector is $\epsilon_e^2=1-\rho^2$.
Therefore, both $\rho$ and $\epsilon_e$ are determined by the
normalized Doppler frequency $f_dT_s$.

The channel in \eqref{ChDelay} is widely-used to model the
time-varying channel. For example, it is used to investigate the
impact of feedback delay on the performance of closed-loop transmit
diversity in \cite{OngGat01Tcomm} and the system capacity and bit
error rate of point-to-point MIMO link in \cite{NguAnd04PIMRC}. It
simplifies the analysis, and the results can be easily extended to
other scenarios. Essentially, this model is of the form
\begin{equation}
\mathbf{h}[n]=\mathbf{g}[n]+\mathbf{e}[n],
\end{equation}
where $\mathbf{g}[n]$ is the available CSI at time $n$ with an
uncorrelated error vector $\mathbf{e}[n]$,
$\mathbf{g}[n]\sim\mathcal{CN}(\mathbf{0},(1-\epsilon_e^2)\mathbf{I})$,
and
$\mathbf{e}[n]\sim\mathcal{CN}(\mathbf{0},\epsilon_e^2\mathbf{I})$.
It can be used to consider the effect of other imperfect CSIT, such
as estimation error and analog feedback. The difference is in
$\mathbf{e}[n]$, which has different variance $\epsilon_e^2$ for
different scenarios. Some examples are given as follows.

\paragraph{Estimation Error} If the receiver obtains the CSI
through MMSE estimation from $\tau_p$ pilot symbols, the error
variance is $\epsilon_e^2=\frac{1}{1+\tau_p\gamma_p}$, where
$\gamma_p$ is the SNR of the pilot symbol \cite{Poor94}.

\paragraph{Analog Feedback} For analog feedback, the error variance is $\epsilon_e^2=\frac{1}{1+\tau_{ul}\gamma_{ul}}$,
where $\tau_{ul}$ is the number of channel uses per channel
coefficient and $\gamma_{ul}$ is the SNR on the uplink feedback
channel \cite{CaiJin07ISIT}.

\paragraph{Analog Feedback with Prediction} As shown in \cite{KobCai07JSAC}, for analog feedback
with a $d$-step MMSE predictor and the Gauss-Markov model, the error
variance is
$\epsilon_e^2=\rho^{2d}\epsilon_0+(1-\rho^2)\sum_{l=0}^{d-1}\rho^{2l}$,
where $\rho$ is the same as in \eqref{ChDelay} and $\epsilon_0$ is
the Kalman filtering mean-square error.

Therefore, the results in this paper can be easily extended to these
systems. In the following parts, we focus on the effect of CSI
delay.

\subsection{Channel Quantization Model} We consider frequency-division
duplexing (FDD) systems, where limited feedback techniques provide
partial CSIT through a dedicated feedback channel from the receiver
to the transmitter. The channel direction information for the
precoder design is fed back using a quantization codebook known at
both the transmitter and receiver.

The quantization is chosen from a codebook of unit norm vectors of
size $L=2^B$. We assume each user uses a different codebook to avoid
the same quantization vector. The codebook for user $u$ is
$\mathcal{C}_u=\{\mathbf{c}_{u,1},\mathbf{c}_{u,2},\cdots,\mathbf{c}_{u,L}\}$.
Each user quantizes its channel to the closest codeword, where
closeness is measured by the inner product. Therefore, the index of
channel for user $u$ is
\begin{equation}\label{eq_Q}
I_u=\arg\max_{1\leq{\ell}\leq{L}}|\tilde{\mathbf{h}}_u^*\mathbf{c}_{u,\ell}|.
\end{equation}
Each user needs to feed back $B$ bits to denote this index, and the
transmitter has the quantized channel information
$\hat{\mathbf{h}}_u=\mathbf{c}_{u,I_u}$. As the optimal vector
quantizer for this problem is not known in general, random vector
quantization (RVQ) \cite{SanHon04ISIT} is used, where each
quantization vector is independently chosen from the isotropic
distribution on the $N_t$-dimensional unit sphere. It has been shown
in \cite{Jin06IT} that RVQ can facilitate the analysis and provide
performance close to the optimal quantization. In this paper, we
analyze the achievable rate averaged over both RVQ-based random
codebooks and fading distributions.

An important metric for the limited feedback system is the squared
angular distortion, defined as
$\sin^2\left(\theta_u\right)=1-|\tilde{\mathbf{h}}_u^*\hat{\mathbf{h}}_u|^2$,
where
$\theta_u=\angle\left(\tilde{\mathbf{h}}_u,\hat{\mathbf{h}}_u\right)$.
With RVQ, it was shown in \cite{Jin06IT,AuLov07Twc} that the
expectation in i.i.d. Rayleigh fading is given by
\begin{equation}
\mathbb{E}_\theta\left[\sin^2\left(\theta_u\right)\right]=2^B\cdot\beta\left(2^B,\frac{N_t}{N_t-1}\right),
\end{equation}
where $\beta(\cdot)$ is the beta function. It can be tightly bounded
as \cite{Jin06IT}
\begin{equation}\label{eqnhh}
\frac{N_t-1}{N_t}2^{-\frac{B}{N_t-1}}\leq\mathbb{E}\left[\sin^2\left(\theta_u\right)\right]\leq{2^{-\frac{B}{N_t-1}}}.
\end{equation}

\section{Transmission Techniques}\label{Sec:TransTech}
In this section, we describe the transmission techniques for both SU
and MU MIMO systems with perfect CSIT, which will be used in the
subsequent sections for imperfect CSIT systems. By doing this, we
focus on the impacts of imperfect CSIT on the conventional
transmission techniques. Designing imperfect CSIT-aware precoders is
left to future work. Throughout this paper, we use the achievable
ergodic rate as the performance metric for both SU and MU-MIMO
systems. The base station transmits to a single user ($U=1$) for the
SU-MIMO system and to $N_t$ users ($U=N_t$) for the MU-MIMO system.
The SU/MU mode switching algorithm is also described.

\subsection{SU-MIMO System}
With perfect CSIT, it is optimal for the SU-MIMO system to transmit
along the channel direction \cite{Telatar99}, i.e. selecting the
beamforming (BF) vector as $\mathbf{f}[n]=\tilde{\mathbf{h}}[n]$,
denoted as \emph{eigen-beamforming} in this paper. The ergodic
capacity of this system is the same as that of a maximal ratio
combining diversity system, given by \cite{AloGol99Tvt}
\begin{equation}\label{C_BF}
R_{BF}(P)=\mathbb{E}_{\mathbf{h}}\left[\log_2\left(1+P\|\mathbf{h}[n]\|^2\right)\right]=\log_2(e)e^{1/P}\sum_{k=0}^{N_t-1}\frac{\Gamma(-k,1/P)}{P^k},
\end{equation}
where $\Gamma(\cdot,\cdot)$ is the complementary incomplete gamma
function defined as
$\Gamma(\alpha,x)=\int_x^\infty{t^{\alpha-1}e^{-t}dt}$.

\subsection{MU-MIMO System}
For MIMO broadcast channels, although dirty-paper coding (DPC)
\cite{Costa83} is optimal
\cite{CaiSha03IT,Yu04Cioffi,Vishwanath03,Vishwanath03Tse,WeiSte06IT},
it is difficult to implement in practice. As in
\cite{Jin06IT,CaiJin07Submit}, ZF precoding is used in this paper,
which is a linear precoding technique that precancels inter-user
interference at the transmitter. There are several reasons for us to
use this simple transmission technique. Firstly, due to its simple
structure, it is possible to derive closed-form results, which can
provide helpful insights. Second, the ZF precoding is able to
provide full spatial multiplexing gain and only has a power offset
compared to the optimal DPC system \cite{Jin05ISIT}. In addition, it
was shown in \cite{Jin05ISIT} that the ZF precoding is optimal among
the set of all linear precoders at asymptotically high SNR. In
Section \ref{Sec:Num}, we will show that our results for the ZF
system also apply for the regularized ZF precoding
\cite{PeeHoc05Tcomm}, which provides a higher throughput than the ZF
precoding at low to moderate SNRs.

With precoding vectors $\mathbf{f}_u[n], u=1,2,\cdots, U,$ assuming
equal power allocation\footnote{At high SNR, this performs closely
to the system employing optimal water-filling, as power allocation
mainly benefits at low SNR.}, the received SINR for the $u$-th user
is given as
\begin{equation}\notag
\gamma_{ZF,u}=\frac{\frac{P}{U}|\mathbf{h}^*_u[n]\mathbf{f}_u[n]|^2}{1+\frac{P}{U}\sum_{u'\neq{u}}|\mathbf{h}_u^*[n]\mathbf{f}_{u'}[n]|^2}.
\end{equation}
This is true for a general linear precoding MU-MIMO system. With
perfect CSIT, this quantity can be calculated at the transmitter,
while with imperfect CSIT, it can be estimated at the receiver and
fed back to the transmitter given knowledge of $\mathbf{f}_u[n]$.

Denote
$\tilde{\mathbf{H}}[n]=[\tilde{\mathbf{h}}_1[n],\tilde{\mathbf{h}}_2[n],\cdots,\tilde{\mathbf{h}}_U[n]]^*$.
With perfect CSIT, the ZF precoding vectors are determined from the
pseudo-inverse of $\tilde{\mathbf{H}}[n]$, as
$\mathbf{F}[n]=\tilde{\mathbf{H}}^\dag[n]=\tilde{\mathbf{H}}^*[n](\tilde{\mathbf{H}}[n]\tilde{\mathbf{H}}^*[n])^{-1}$.
The precoding vector for the $u$-th user is obtained by normalizing
the $u$-th column of $\mathbf{F}[n]$. Therefore,
$\mathbf{h}^*_u[n]\mathbf{f}_{u'}[n]=0,\,\forall{u\neq{u'}}$, i.e.
there is no inter-user interference. The received SINR for the
$u$-th user becomes
\begin{equation}\label{SINR_ZF}
\gamma_{ZF,u}={\frac{P}{U}|\mathbf{h}^*_u[n]\mathbf{f}_u[n]|^2}.
\end{equation}
As $\mathbf{f}_u[n]$ is independent of $\mathbf{h}_u[n]$, and
$\|\mathbf{f}_u[n]\|^2=1$, the effective channel for the $u$-th user
is a single-input single-output (SISO) Rayleigh fading channel.
Therefore, the achievable sum rate for the ZF system is given by
\begin{equation}\label{R_ZF}
R_{ZF}(P)=\sum_{u=1}^U\mathbb{E}_\gamma\left[\log_2(1+\gamma_{ZF,u})\right].
\end{equation}
Each term on the right hand side of \eqref{R_ZF} is the ergodic
capacity of a SISO system in Rayleigh fading, given in
\cite{AloGol99Tvt} as
\begin{equation}\label{R_ZFi}
R_{ZF,u}=\mathbb{E}_\gamma\left[\log_2(1+\gamma_{ZF,u})\right]=\log_2(e)e^{U/P}E_1(U/P),
\end{equation}
where $E_1(\cdot)$ is the exponential-integral function of the first
order, $E_1(x)=\int_1^\infty\frac{e^{-xt}}{t}dt$.

\subsection{SU/MU Mode Switching}
Imperfect CSIT will degrade the performance of the MIMO
communication. In this case, it is unclear whether and when the
MU-MIMO system can actually provide a throughput gain over the
SU-MIMO system. Based on the analysis of the achievable ergodic
rates in this paper, we propose to switch between SU and MU modes
and select the one with the higher achievable rate.

The channel correlation coefficient $\rho$, which captures the CSI
delay effect, usually varies slowly. The quantization codebook size
is normally fixed for a given system. Therefore, it is reasonable to
assume that the transmitter has knowledge of both delay and channel
quantization, and can estimate the achievable ergodic rates of both
SU and MU MIMO modes. Then it can determine the active mode and
select one (SU mode) or $N_t$ (MU mode) users to serve. This is a
low-complexity transmission strategy, and can be combined with
random user selection, round-robin scheduling, or scheduling based
on queue length rather than channel status. It only requires the
selected users to feed back instantaneous channel information.
Therefore, it is suitable for a system that has a constraint on the
total feedback bits and only allows a small number of users to send
feedback, or a system with a strict delay constraint that cannot
employ opportunistic scheduling based on instantaneous channel
information.

To determine the transmission rate, the transmitter sends pilot
symbols, from which the active users estimate the received SINRs and
feed back them to the transmitter. In this paper, we assume the
transmitter knows perfectly the actual received SINR at each active
user. In practice, there will inevitably be errors in such
information due to estimation error and feedback delay, which will
result in rate mismatch, i.e. the transmission rate based on the
estimated SINR does not match the actual SINR on the channel, so
there will be outage events. How to deal with such rate mismatch is
of practical importance, and we mention several possible approaches
as follows. The full investigation of this issue requires further
research and is out of scope of this paper. Considering the outage
events, the transmission strategy can be designed based on the
actual information symbols successfully delivered to the receiver,
denoted as \emph{goodput} in \cite{LauJia06Tvt,WuLau08Twc}. With the
estimated SINR, another approach is to back off on the transmission
rate based on the variance of the estimation error, as did in
\cite{VakSha06ICASSP,VakHas06SPAWC} for the single-antenna
opportunistic scheduling system and in \cite{VakDan07IWCMC} for the
multiple-antenna opportunistic beamforming system. Combined with
user selection, the transmission rate can also be determined based
on some lower bound of the actual SINR to make sure that no outage
occurs, as did in \cite{KouFra07ICASSP} for the limited feedback
system.

\section{Performance Analysis and Mode Switching}\label{Sec:ModSel}
In this section, we investigate the achievable ergodic rates for
both SU and MU MIMO modes. We first analyze the average received SNR
for the BF system and the average residual interference for the ZF
system, which provide insights on the impact of imperfect CSIT. To
select the active mode, accurate closed-form approximations for both
SU and MU modes are then derived.

\subsection{SU Mode--Eigen-Beamforming}
First, if there is no delay and only channel quantization, the BF
vector is based on the quantized feedback,
$\mathbf{f}^{(Q)}[n]=\hat{\mathbf{h}}[n]$. The average received SNR
is
\begin{align}
\overline{\mbox{SNR}}_{BF}^{(Q)}&=\mathbb{E}_{\mathbf{h},\mathcal{C}}[P|\mathbf{h}^*[n]\hat{\mathbf{h}}[n]|^2]\notag\\
&=\mathbb{E}_{\mathbf{h},\mathcal{C}}[P\|\mathbf{h}[n]\|^2|\tilde{\mathbf{h}}^*[n]\hat{\mathbf{h}}[n]|^2]\notag\\
&\stackrel{(a)}{\leq}P{N_t}\left(1-\frac{N_t-1}{N_t}2^{-\frac{B}{N_t-1}}\right),\label{BF_Q}
\end{align}
where (a) follows the independence between $\|\mathbf{h}[n]\|^2$ and
$|\tilde{\mathbf{h}}^*[n]\hat{\mathbf{h}}[n]|^2$, together with the
result in \eqref{eqnhh}.

With both delay and channel quantization, the BF vector is based on
the quantized channel direction with delay, i.e.
$\mathbf{f}^{(QD)}[n]=\hat{\mathbf{h}}[n-1]$. The instantaneous
received SNR for the BF system
\begin{align}
{\mbox{SNR}}_{BF}^{(QD)}&=P\Big|\mathbf{h}^*[n]{\mathbf{f}^{(QD)}[n]}\Big|^2.
\end{align}

Based on \eqref{BF_Q}, we get the following theorem on the average
received SNR for the SU mode.

\begin{theorem}\label{thmBFQD}
The average received SNR for a BF system with channel quantization
and CSI delay is
\begin{equation}
\overline{\mbox{SNR}}_{BF}^{(QD)}\leq{P}{N_t}\left(\rho^2\Delta_{BF}^{(Q)}+\Delta_{BF}^{(D)}\right)\label{Delta_BF},
\end{equation}
where $\Delta_{BF}^{(Q)}$ and $\Delta_{BF}^{(D)}$ show the impact of
channel quantization and feedback delay, respectively, given by
\begin{equation}
\Delta_{BF}^{(Q)}=1-\frac{N_t-1}{N_t}2^{-\frac{B}{N_t-1}},\quad\Delta_{BF}^{(D)}=\frac{\epsilon_e^2}{N_t}\notag.
\end{equation}
\end{theorem}

\begin{proof}
See Appendix \ref{pthmBFQD}.
\end{proof}

From Jensen's inequality, an upper bound of the achievable rate for
the BF system with both quantization and delay is given by
\begin{align}
R_{BF}^{(QD)}&=\mathbb{E}_{\mathbf{h},\mathcal{C}}\left[\log_2\left(1+{\mbox{SNR}}_{BF}^{(QD)}\right)\right]\notag\\
&\leq\log_2\left[1+\overline{\mbox{SNR}}_{BF}^{(QD)}\right]\notag\\
&\leq\log_2\left[1+P{N_t}\left(\rho^2\Delta_{BF}^{(Q)}+\Delta_{BF}^{(D)}\right)\right]\label{R_BFQD}.
\end{align}

\begin{remark}
Note that $\rho^2=1-\epsilon_e^2$, so the average SNR decreases with
$\epsilon_e^2$. With a fixed $B$ and fixed delay, the SNR
degradation is a constant factor independent of $P$. At high SNR,
the imperfect CSIT introduces a constant rate loss
$\log_2\left(\rho^2\Delta_{BF}^{(Q)}+\Delta_{BF}^{(D)}\right)$.
\end{remark}

The upper bound provided by Jensen's inequality is not tight. To get
a better approximation for the achievable rate, we first make the
following approximation on the instantaneous received SNR
\begin{align}
\mbox{SNR}_{BF}^{(QD)}&=P|\mathbf{h}^*[n]\hat{\mathbf{h}}[n-1]|^2\notag\\
&=P|(\rho\mathbf{h}[n-1]+\mathbf{e}[n])^*\hat{\mathbf{h}}[n-1]|^2\notag\\
&\approx{P}\rho^2|\mathbf{h}^*[n-1]\hat{\mathbf{h}}[n-1]|^2,
\end{align}
i.e. we remove the term with $\mathbf{e}[n]$ as it is normally
insignificant compared to $\rho\mathbf{h}[n-1]$. This will be
verified later by simulation. In this way, the system is
approximated as the one with limited feedback and with equivalent
SNR $\rho^2P$.

From \cite{AuLov07Twc}, the achievable rate of the limited feedback
BF system is given by
\begin{align}
R_{BF}^{(Q)}(P)&=\log_2{(e)}\left(e^{1/P}\sum_{k=0}^{N_t-1}E_{k+1}\left(\frac{1}{P}\right)\right.\notag\\
&\left.-\int_0^1\left(1-(1-x)^{N_t-1}\right)^{2^B}\frac{N_t}{x}e^{1/P{x}}E_{N_t+1}\left(\frac{1}{P{x}}\right)dx\right),
\end{align}
where $E_n(x)=\int_1^\infty{e}^{-xt}x^{-n}dt$ is the $n$-th order
exponential integral. So $R_{BF}^{(QD)}$ can be approximated as
\begin{equation}\label{Approx_BFQD}
R_{BF}^{(QD)}(P)\approx{R}_{BF}^{(Q)}(\rho^2P).
\end{equation}

As a special case, considering a system with delay only, e.g. the
time-division duplexing (TDD) system which can estimate the CSI from
the uplink with channel reciprocity but with propagation and
processing delay, the BF vector is based on the delayed channel
direction, i.e. $\mathbf{f}^{(D)}[n]=\tilde{\mathbf{h}}[n-1]$. We
provide a good approximation for the achievable rate for such a
system as follows.

The instantaneous received SNR is given as
\begin{align}
\mbox{SNR}_{BF}^{(D)}&=P|\mathbf{h}^*[n]\mathbf{f}^{(D)}[n]|^2\notag\\
&=P|(\rho\mathbf{h}[n-1]+\mathbf{e}[n])^*\tilde{\mathbf{h}}[n-1]|^2\notag\\
&\stackrel{(a)}{\approx}{P}\rho^2\|\mathbf{h}[n-1]\|^2+P|\mathbf{e}^*[n]\tilde{\mathbf{h}}[n-1]|^2.
\end{align}
In step (a) we eliminate the cross terms since $\mathbf{e}[n]$ is
normally small. As $\mathbf{e}[n]$ is independent of
$\tilde{\mathbf{h}}[n-1]$,
$\mathbf{e}[n]\sim\mathcal{CN}(\mathbf{0},\epsilon_e^2\mathbf{I})$
and $\|\tilde{\mathbf{h}}[n-1]\|^2=1$, we have
$|\mathbf{e}^*[n]\tilde{\mathbf{h}}[n-1]|^2\sim\chi_2^2$, where
$\chi_M^2$ denotes chi-square distribution with $M$ degrees of
freedom. In addition, $\|\mathbf{h}[n-1]\|^2\sim\chi^2_{2N_t}$, and
it is independent of $|\mathbf{e}^*[n]\tilde{\mathbf{h}}[n-1]|^2$.
Then the following theorem can be derived.
\begin{theorem}\label{thmR_BFD}
The achievable ergodic rate of the BF system with delay can be
approximated as
\begin{align}\label{R_BFD}
R_{BF}^{(D)}&\approx\log_2{(e)}{a_0}^{N_t}e^{1/\eta_2}E_1\left(\frac{1}{\eta_2}\right)\notag\\
&\quad-\log_2{(e)}(1-a_0)\sum_{i=0}^{N_t-1}\sum_{l=0}^i\frac{a_0^{N_t-1-i}}{(i-l)!}\eta_1^{-(i-l)}I_1(1/\eta_1,1,i-l),
\end{align}
where $\eta_1=P\rho^2$, $\eta_2=P\epsilon_e^2$,
$a_0=\frac{\eta_2}{\eta_2-\eta_1}$, and $I_1(\cdot,\cdot,\cdot)$ is
given in \eqref{I1} in Appendix \ref{Result}.
\end{theorem}

\begin{proof}
See Appendix \ref{pthmR_BFD}.
\end{proof}

\subsection{Zero-Forcing}
\subsubsection{Average Residual Interference}
If there is no delay and only channel quantization, the precoding
vectors for the ZF system are designed based on
$\hat{\mathbf{h}}_1[n], \hat{\mathbf{h}}_2[n], \cdots,
\hat{\mathbf{h}}_U[n]$ to achieve
$\hat{\mathbf{h}}^*_u[n]\mathbf{f}^{(Q)}_{u'}[n]=0,\,\forall{u\neq{u'}}$.
With random vector quantization, it is shown in \cite{Jin06IT} that
the average noise plus interference for each user is
\begin{equation}
\Delta^{(Q)}_{ZF,u}=\mathbb{E}_{\mathbf{h},\mathcal{C}}\left[1+\frac{P}{U}\sum_{{u'}\neq{u}}|\mathbf{h}_u^*[n]\mathbf{f}^{(Q)}_{u'}[n]|^2\right]=1+2^{-\frac{B}{N_t-1}}P\label{ZF_Q}.
\end{equation}

With both channel quantization and CSI delay, precoding vectors are
designed based on $\hat{\mathbf{h}}_1[n-1], \hat{\mathbf{h}}_2[n-1],
\cdots, \hat{\mathbf{h}}_U[n-1]$ and achieve
$\hat{\mathbf{h}}^*_u[n-1]\mathbf{f}^{(QD)}_{u'}[n]=0,\,\forall{u\neq{u'}}$.
The received SINR for the $u$-th user is given as
\begin{equation}\label{ZFSINR_QD}
\gamma_{ZF,u}^{(QD)}=\frac{\frac{P}{U}|\mathbf{h}^*_u[n]\mathbf{f}^{(QD)}_u[n]|^2}{1+\frac{P}{U}\sum_{u'\neq{u}}|\mathbf{h}_u^*[n]\mathbf{f}^{(QD)}_{u'}[n]|^2}.
\end{equation}
As $\mathbf{f}^{(QD)}_u[n]$ is in the nullspace of
$\hat{\mathbf{h}}_{u'}[n-1]$ $\forall{u'\neq{u}}$, it is
isotropically distributed in $\mathbb{C}^{N_t}$ and independent of
$\tilde{\mathbf{h}}_u[n-1]$ as well as $\tilde{\mathbf{h}}_u[n]$, so
$|\mathbf{h}^*_u[n]\mathbf{f}^{(QD)}_u[n]|^2\sim\chi_2^2$. The
average noise plus interference is given in the following theorem.
\begin{theorem}\label{thmSINRZF_QD}
The average noise plus interference for the $u$-th user of the ZF
system with both channel quantization and CSI delay is
\begin{equation}
\Delta_{ZF,u}^{(QD)}=1+(U-1)\frac{P}{U}\left(\rho_u^2\Delta_{ZF,u}^{(Q)}+\Delta_{ZF,u}^{(D)}\right),
\end{equation}
where $\Delta_{ZF,u}^{(Q)}$ and $\Delta_{ZF,u}^{(D)}$ are the
degradations brought by channel quantization and feedback delay,
respectively, given by
\begin{equation}
\Delta_{ZF,u}^{(Q)}=\frac{U}{U-1}2^{-\frac{B}{N_t-1}},\quad
\Delta_{ZF,u}^{(D)}=\epsilon_{e,u}^2\notag.
\end{equation}
\end{theorem}

\begin{proof}
The proof is similar to the one for \emph{Theorem \ref{thmBFQD}} in
appendix \ref{pthmBFQD}.
\end{proof}

\begin{remark}
From \emph{Theorem \ref{thmSINRZF_QD}} we see that the average
residual interference for a given user consists of three parts:
\renewcommand{\labelenumi}{(\roman{enumi})}
\begin{enumerate}
\item \emph{The number of interferers}, $U-1$. The more users the system supports, the
higher the mutual interference.
\item \emph{The transmit power of the other active users}, $\frac{P}{U}$.
As the transmit power increases, the system becomes
interference-limited. It is possible to improve performance through
power allocation, which is left to future work.
\item \emph{The CSIT accuracy for this user}, which is reflected from $\rho_u^2\Delta_{ZF,u}^{(Q)}+\Delta_{ZF,u}^{(D)}$. The
user with a larger delay or a smaller codebook size suffers a higher
residual interference.
\end{enumerate}
\end{remark}

From this remark, the interference term,
$\frac{P}{U}(U-1)\epsilon_{e,u}^2$, equivalently comes from $U-1$
\emph{virtual interfering users}, each with equivalent SNR as
$\frac{P}{U}\left(\rho_u^2\Delta_{ZF,u}^{(Q)}+\Delta_{ZF,u}^{(D)}\right)$.
With a high $P$ and a fixed $\epsilon_{e,u}$ or $B$, the system is
interference-limited and cannot achieve full spatial multiplexing
gain. Therefore, to keep a constant rate loss, i.e. to sustain the
spatial multiplexing gain, the channel error due to both
quantization and delay needs to be reduced as SNR increases. Similar
to the result for the limited feedback system in \cite{Jin06IT}, for
the ZF system with both delay and channel quantization, we can get
the following corollary for the condition to achieve the full
spatial multiplexing gain.
\begin{corollary}\label{CvBvD}
To keep a constant rate loss of $\log_2\delta_0$ bps/Hz for each
user, the codebook size and CSI delay need to satisfy the following
condition
\begin{equation}\label{eqnvBvD}
\rho_u^2\Delta_{ZF,u}^{(Q)}+\Delta_{ZF,u}^{(D)}=\frac{U}{U-1}\cdot\frac{\delta_0-1}{P}.
\end{equation}
\end{corollary}
\begin{proof}
As shown in \cite{Jin06IT,CaiJin07Submit}, the rate loss for each
user due to imperfect CSIT is upper bounded by
$\Delta{R}_u\leq\log_2\Delta_{ZF,u}^{(QD)}$. The corollary follows
from solving $\log_2\Delta_{ZF,u}^{(QD)}=\log_2\delta_0$.
\end{proof}
Equivalently, this means that for a given $\rho^2$, the feedback
bits per user needs to scale as
\begin{equation}
B=(N_t-1)\log_2\left(\frac{\delta_0-1}{\rho_u^2P}-\frac{U-1}{U}\cdot\left(\frac{1}{\rho_u^2}-1\right)\right)^{-1}.
\end{equation}
As $\rho_u^2\rightarrow{1}$, i.e. there is no CSI delay, the
condition becomes $B=(N_t-1)\log_2\frac{P}{\delta_0-1}$, which
agrees with the result in \cite{Jin06IT} with limited feedback only.

\subsubsection{Achievable Rate}
For the ZF system with imperfect CSI, the genie-aided upper bound
for the ergodic achievable rate\footnote{This upper bound is
achievable only when a genie provides users with perfect knowledge
of all interference and the transmitter knows perfectly the received
SINR at each user.} is given by \cite{CaiJin07Submit}
\begin{equation}\label{R_ZFub}
R_{ZF}^{(QD)}\leq\sum_{u=1}^U\mathbb{E}_\gamma\left[\log_2\left(1+\gamma_{ZF,u}^{(QD)}\right)\right]=R_{ZF,ub}^{(QD)}.
\end{equation}
We assume the mobile users can perfectly estimate the noise and
interference and feed back it to the transmitter, so the upper bound
is chosen as the performance metric, i.e.
$R_{ZF}^{(QD)}=R_{ZF,ub}^{(QD)}$, as in
\cite{Jin06IT,DinLov07Tsp,YooJin07JSAC}.

The following lower bound based on the rate loss analysis is used in
\cite{Jin06IT,CaiJin07Submit}
\begin{equation}\label{R_ZFlb}
R_{ZF}^{(QD)}\geq{R_{ZF}}-\sum_{u=1}^U\log_2\Delta_{ZF,u}^{(QD)},
\end{equation}
where $R_{ZF}$ is the achievable rate with perfect CSIT, given in
\eqref{R_ZF}. However, this lower bound is very loose. In the
following, we will derive a more accurate approximation for the
achievable rate for the ZF system.

To get a good approximation for the achievable rate for the ZF
system, we first approximate the instantaneous SINR as
\begin{align}\label{ZFSINR_QD_approx}
\gamma_{ZF,u}^{(QD)}&=\frac{\frac{P}{U}|\mathbf{h}^*_u[n]\mathbf{f}^{(QD)}_u[n]|^2}{1+\frac{P}{U}\sum_{u'\neq{u}}|(\rho_u\mathbf{h}_u[n-1]+\mathbf{e}_u[n])^*\mathbf{f}^{(QD)}_{u'}[n]|^2}\notag\\
&\approx\frac{\frac{P}{U}|\mathbf{h}^*_u[n]\mathbf{f}^{(QD)}_u[n]|^2}{1+\frac{P}{U}\left(\sum_{u'\neq{u}}\rho_u^2|\mathbf{h}_u^*[n-1]\mathbf{f}^{(QD)}_{u'}[n]|^2
+\sum_{u'\neq{u}}|\mathbf{e}_u^*[n]\mathbf{f}^{(QD)}_{u'}[n]|^2\right)},
\end{align}
i.e. eliminating the interference terms which have both
$\mathbf{h}_u[n-1]$ and $\mathbf{e}_u[n]$ as $\mathbf{e}_u[n]$ is
normally very small, so we get two separate interference sums due to
delay and quantization, respectively.

For the interference term due to delay,
$|\mathbf{e}_u^*[n]\mathbf{f}^{(QD)}_{u'}[n]|^2\sim\chi_2^2$, as
$\mathbf{e}[n]$ is independent of ${\mathbf{f}}^{(QD)}_{u'}[n]$ and
$\|{\mathbf{f}}^{(QD)}_{u'}[n]\|^2=1$. For the interference term due
to quantization, it was shown in \cite{Jin06IT} that
$|\tilde{\mathbf{h}}_u^*[n-1]{\mathbf{f}}_{u'}^{(QD)}[n]|^2$ is
equivalent to the product of the quantization error $\sin^2\theta_u$
and an independent $\beta(1,N_t-2)$ random variable. Therefore, we
have
\begin{equation}
|\mathbf{h}_u^*[n-1]\mathbf{f}^{(QD)}_{u'}[n]=\|\mathbf{h}_u[n-1]\|^2(\sin^2\theta_u)\cdot\beta(1,N_t-2).
\end{equation}
In \cite{YooJin07JSAC}, with a quantization cell
approximation\footnote{The quantization cell approximation is based
on the ideal assumption that each quantization cell is a Voronoi
region on a spherical cap with the surface area $2^{-B}$ of the
total area of the unit sphere for a $B$ bits codebook. The detail
can be found in \cite{MukSab03IT,ZhoWan05Twc,YooJin07JSAC}.}
\cite{MukSab03IT,ZhoWan05Twc}, it was shown that
$\|\mathbf{h}_u[n-1]\|^2(\sin^2\theta_u)$ has a Gamma distribution
with parameters $(N_t-1,\delta)$, where
$\delta=2^{-\frac{B}{N_t-1}}$. As shown in \cite{YooJin07JSAC} the
analysis based on the quantization cell approximation is close to
the performance of random vector quantization, so we use this
approach to derive the achievable rate.

The following lemma gives the distribution of the interference term
due to quantization.
\begin{lemma}\label{lemma_Iq}
Based on the quantization cell approximation, the interference term
due to quantization in \eqref{ZFSINR_QD_approx},
$|\mathbf{h}_u[n-1]\mathbf{f}^{(QD)}_{u'}[n]|^2$, is an exponential
random variable with mean $\delta$, i.e. its probability
distribution function (pdf) is
\begin{equation}
p(x)=\frac{1}{\delta}e^{-x/\delta}, x\geq{0}.
\end{equation}
\end{lemma}

\begin{proof}
See Appendix \ref{plemma_Iq}.
\end{proof}

\renewcommand{\labelenumi}{(\roman{enumi})}
\begin{remark}
From this lemma, we see that the residual interference terms due to
both delay and quantization are exponential random variables, which
means the delay and quantization error have equivalent effects, only
with different means. By comparing the means of these two terms,
i.e. comparing $\epsilon_e^2$ and $2^{-\frac{B}{N_t-1}}$, we can
find the dominant one. In addition, with this result, we can
approximate the achievable rate of the ZF limited feedback system,
which will be provided later in this section.
\end{remark}

Based on the distribution of the interference terms, the
approximation for the achievable rate for the MU mode is given in
the following theorem.
\begin{theorem}\label{thm_RQD}
The ergodic achievable rate for the $u$-th user in the MU mode with
both delay and channel quantization can be approximated as
\begin{align}\label{R_QD}
R_{ZF,u}^{(QD)}\approx\log_2(e)\sum_{i=0}^{M-1}\sum_{j=1}^2\left[a^{(j)}_ii!\left(\frac{\alpha}{\beta}\right)^{i+1}\cdot{I_3}\left(\frac{1}{\alpha},\frac{\alpha}{\beta\delta_j},i+1\right)\right],
\end{align}
where $\alpha=\beta=\frac{P}{U}$, $\delta_1=\rho_u^2\delta$,
$\delta_2=\epsilon_{e,u}^2$, $M=N_t-1$, $a^{(1)}_i$ and $a^{(2)}_i$
are given in \eqref{eq_a1} and \eqref{eq_a2}, and
$I_3(\cdot,\cdot,\cdot)$ is given in \eqref{I3} in Appendix
\ref{Result}.
\end{theorem}
\begin{proof}
See Appendix \ref{pthm_RQD}.
\end{proof}
The ergodic sum throughput is
\begin{equation}\label{sr_QD}
R_{ZF}^{(QD)}=\sum_{u=1}^{U}R_{ZF,u}^{(QD)}.
\end{equation}

As a special case, for a ZF system with delay only, we can get the
following approximation for the ergodic achievable rate.
\begin{corollary}
The ergodic achievable rate for the $u$-th user in the ZF system
with delay is approximated as
\begin{align}\label{R_D}
R_{ZF,u}^{(D)}\approx\log_2(e)\left(\frac{\alpha}{\beta}\right)^{-(M-1)}\cdot{I_3}\left(\frac{1}{\alpha},\frac{\alpha}{\beta},M-1\right),
\end{align}
where $\alpha=\frac{P}{U}$, $\beta=\frac{\epsilon_{e,u}^2P}{U}$,
$M=N_t-1$, and $I_3(\cdot,\cdot,\cdot)$ is given in \eqref{I3} in
Appendix \ref{Result}.
\end{corollary}
\begin{proof}
Following the same steps in Appendix \ref{pthm_RQD} with
$\delta_1=0$.
\end{proof}
\begin{remark}
As shown in \emph{Lemma \ref{lemma_Iq}}, the effects of delay and
channel quantization are equivalent, so the approximation in
\eqref{R_D} also applies for the limited feedback system. This is
verified by simulation in Fig. \ref{fig_ZFLFB}, which shows that
this approximation is very accurate and can be used to analyze the
limited feedback system.
\end{remark}

\subsection{Mode Switching}\label{MS_QD}
We first verify the approximation \eqref{R_QD} in Fig.
\ref{fig_ZFSimvsCal}, which compares the approximation with
simulation results and the lower bound \eqref{R_ZFlb}, with $B=10$,
$v=20$ km/hr, $f_c=2$ GHz, and $T_s=1$ msec. We see that the lower
bound is very loose, while the approximation is accurate especially
for $N_t=2$. In fact, the approximation turns out to be a lower
bound. Note that due to the imperfect CSIT, the sum rate reduces
with $N_t$.

In Fig. \ref{figBFZFswitch_QD}, we compare the BF and ZF systems,
with $B=18$, $f_c=2$ GHz, $v=10$ km/hr, and $T_s=1$ msec. We see
that the approximation for the BF system almost matches the
simulation exactly. The approximation for the ZF system is accurate
at low to medium SNRs, and becomes a lower bound at high SNR, which
is approximately $0.7$ bps/Hz in total, or $0.175$ bps/Hz per user,
lower than the simulation. The throughput of the ZF system is
limited by the residual inter-user interference at high SNR, where
it is lower than the BF system. This motivates to switch between the
SU and MU MIMO modes. The approximations \eqref{Approx_BFQD} and
\eqref{R_QD} will be used to calculate the mode switching points.
There may be two switching points for the system with delay, as the
SU mode will be selected at both low and high SNR. These two points
can be calculated by providing different initial values to the
nonlinear equation solver, such as \emph{fsolve} in MATLAB.

\section{Numerical Results}\label{Sec:Num}
In this section, numerical results are presented. First, the
operating regions for different modes are plotted, which show the
impact of different parameters, including the normalized Doppler
frequency, the codebook size, and the number of transmit antennas.
Then the extension of our results for the ZF precoding to the MMSE
precoding is demonstrated.

\subsection{Operating Regions}
As shown in Section \ref{MS_QD}, finding mode switching points
requires solving a nonlinear equation, which does not have a
closed-form solution and gives little insight. However, it is easy
to evaluate numerically for different parameters, from which
insights can be drawn. In this section, with the calculated mode
switching points for different parameters, we plot the operating
regions for both SU and MU modes. The active mode for the given
parameter and the condition to activate each mode can be found from
such plots.

In Fig. \ref{figBFZFswitchregion_DQ}, the operating regions for both
SU and MU modes are plotted, for different normalized Doppler
frequencies and different number of feedback bits in Fig.
\ref{figBFZFswitchregion_DQ_SNRD} and Fig.
\ref{figBFZFswitchregion_DQ_SNRQ}, respectively, and with $U=N_t=4$.
There are analogies between two plots. Some key observations are as
follows:
\begin{enumerate}
\item For the delay plot Fig. \ref{figBFZFswitchregion_DQ_SNRD}, comparing the two curves for $B=16$
and $B=20$, we see that the smaller the codebook size, the smaller
the operating region for the ZF mode. For the ZF mode to be active,
$f_dT_s$ needs to be small, specifically we need $f_dT_s<0.055$ and
$f_dT_s<0.046$ for $B=20$ and $B=16$, respectively. These conditions
are not easily satisfied in practical systems. For example, with
carrier frequency $f_c=2$ GHz, mobility $v=20$ km/hr, the Doppler
frequency is $37$ Hz, and then to satisfy $f_dT_s<0.055$ the delay
should be less than $1.5$ msec.
\item For the codebook size plot Fig.
\ref{figBFZFswitchregion_DQ_SNRQ}, comparing the two curves with
$v=10$ km/hr and $v=20$ km/hr, as $f_dT_s$ increases ($v$
increases), the ZF operating region shrinks. For the ZF mode to be
active, we should have $B\geq{12}$ and $B\geq{14}$ for $v=10$ km/hr
and $v=20$ km/hr, respectively, which means a large codebook size.
Note that for BF we only need a small codebook size to get the
near-optimal performance \cite{LovHea04CommMag}.
\item For a given $f_dT_s$ and $B$, the SU mode will be active at
both low and high SNRs, which is due to its array gain and the
robustness to imperfect CSIT, respectively.
\end{enumerate}

The operating regions for different $N_t$ are shown in Fig.
\ref{figBFZFswitchregion_DQ_U}. We see that as $N_t$ increases, the
operating region for the MU mode shrinks. Specifically, we need
$B>12$ for $N_t=4$, $B>19$ for $N_t=6$, and $B>26$ for $N_t=8$ to
get the MU mode activated. Note that the minimum required feedback
bits per user for the MU mode grow approximately linearly with
$N_t$.

\subsection{ZF vs. MMSE Precoding}
It is shown in \cite{PeeHoc05Tcomm} that the regularized ZF
precoding, denoted as \emph{MMSE precoding} in this paper, can
significantly increase the throughput at low SNR. In this section,
we show that our results on mode switching with ZF precoding can
also be applied to MMSE precoding.

Denote
$\hat{\mathbf{H}}[n]=\left[\hat{\mathbf{h}}_1[n],\hat{\mathbf{h}}_2[n],\cdots,\hat{\mathbf{h}}_U[n]\right]^*$.
Then the MMSE precoding vectors are chosen to be the normalized
columns of the matrix \cite{PeeHoc05Tcomm}
\begin{equation}
\hat{\mathbf{H}}^*[n]\left(\hat{\mathbf{H}}[n]\hat{\mathbf{H}}^*[n]+\frac{U}{P}\mathbf{I}\right)^{-1}.
\end{equation}
From this, we see that the MMSE precoders converge to ZF precoders
at high SNR. Therefore, our derivations for the ZF system also apply
to the MMSE system at high SNR.

In Fig. \ref{fig_MMSE}, we compare the performance of ZF and MMSE
precoding systems with delay\footnote{This can also be done in the
system with both delay and quantization, which is more
time-consuming. As shown in \emph{Lemma \ref{lemma_Iq}}, the effects
of delay and quantization are equivalent, so the conclusion will be
the same.}. We see that the MMSE precoding outperforms ZF at low to
medium SNRs, and converges to ZF at high SNR while converges to BF
at low SNR. In addition, it has the same rate ceiling as the ZF
system, and crosses the BF curve roughly at the same point, after
which we need to switch to the SU mode. Based on this, we can use
the second predicted mode switching point (the one at higher SNR) of
the ZF system for the MMSE system. We compare the simulation results
and calculation results by \eqref{R_BFD} and \eqref{R_D} for the
mode switching points in Table \ref{MMSE}. For the ZF system, it is
the second switching point; for the MMSE system, it is the only
switching point. We see that the switching points for MMSE and ZF
systems are very close, and the calculated ones are roughly
$2.5\sim{3}$ dB lower.

\section{Conclusions}\label{Sec:Con}
In this paper, we compare the SU and MU MIMO transmissions in the
broadcast channel with delayed and quantized CSIT, where the amount
of delay and the number of feedback bits per user are fixed. The
throughput of MU-MIMO saturates at high SNR due to residual
inter-user interference, for which a SU/MU mode switching algorithm
is proposed. We derive accurate closed-form approximations for the
ergodic rates for both SU and MU modes, which are then used to
calculate the mode switching points. It is shown that the MU mode is
only possible to be active in the medium SNR regime, with a small
normalized Doppler frequency and a large codebook size.

For future work, the MU-MIMO mode studied in this paper is designed
with zero-forcing criterion, which is shown to be sensitive to CSI
imperfections, so robust precoding design is needed and the impact
of the imperfect CSIT on non-linear precoding should be
investigated. As power control is an effective way to combat
interference, it is interesting to consider the efficient power
control algorithm rather than equal power allocation to improve the
performance, especially in the heterogeneous scenario. It is also of
practical importance to investigate possible approaches to improve
the quality of the available CSIT with a fixed codebook size, e.g.
through channel prediction.


\useRomanappendicesfalse
\appendix
\subsection{Useful Results for Rate Analysis}\label{Result}
In this Appendix, we present some useful results that are used for
rate analysis in this paper.

The following lemma will be used frequently in the derivation of the
achievable rate.
\begin{lemma}\label{lemma_rate}
For a random variable $x$ with probability distribution function
(pdf) $f_X(x)$ and cumulative distribution function (cdf) $F_X(x)$,
we have
\begin{equation}
\mathbb{E}_X\left[\ln(1+X)\right]=\int_0^\infty\frac{1-F_X(x)}{1+x}dx.
\end{equation}
\end{lemma}

\begin{proof}
The proof follows the integration by parts.
\begin{align}
\mathbb{E}_X\left[\ln(1+X)\right]&=\int_0^\infty\ln(1+x)f_X(x)dx\notag\\
&=-\int_0^\infty\ln(1+x)\left[1-F_X(x)\right]'dx\notag\\
&\stackrel{(a)}{=}\int_0^\infty\frac{1-F_X(x)}{1+x}dx,
\end{align}
where $g'$ is the derivative of the function $g$, and step (a)
follows the integration by parts.
\end{proof}

The following lemma provides some useful integrals for rate
analysis, which can be derived from the results in \cite{GraRyz94}.
\begin{lemma}\label{lemma_int}
\begin{align}
I_1(a,b,m)&=\int_0^\infty\frac{x^me^{-a{x}}}{x+b}dx=\sum_{k=1}^m(k-1)!(-b)^{m-k}a^{-k}-(-1)^{m-1}b^me^{ab}E_1(ab)\label{I1}\\
I_2(a,b,m)&=\int_0^\infty\frac{e^{-ax}}{(x+b)^m}dx\notag\\
&=\left\{\begin{array}{lc}e^{ab}E_1(ab)&m=1\\\sum_{k=1}^{m-1}\frac{(k-1)!}{(m-1)!}\frac{(-a)^{m-k-1}}{b^k}+\frac{(-a)^{m-1}}{(m-1)!}e^{ab}E_1(ab)&m\geq{2}\end{array}\right.\label{I2}\\
I_3(a,b,m)&=\int_0^\infty\frac{e^{-ax}}{(x+b)^m(x+1)}dx\notag\\
&=\sum_{i=1}^m(-1)^{i-1}(1-b)^{-i}\cdot{I_2}\left(a,b,m-i+1\right)+(b-1)^{-m}\cdot{I_2}\left(a,1,1\right)\label{I3},
\end{align}
where $E_1(x)$ is the exponential-integral function of the first
order.
\end{lemma}

\subsection{Proof of Theorem \ref{thmBFQD}}\label{pthmBFQD}
The average SNR is
\begin{align}
\overline{\mbox{SNR}}_{BF}^{(QD)}&=\mathbb{E}\left[P\Big|\mathbf{h}^*[n]{\mathbf{f}^{(QD)}[n]}\Big|^2\right]\notag\\
&=P\mathbb{E}\left[\Big|(\rho\mathbf{h}[n-1]+\mathbf{e}[n])^*\hat{\mathbf{h}}[n-1]\Big|^2\right]\notag\\
&\stackrel{(a)}{=}P|\rho\mathbf{h}^*[n-1]\hat{\mathbf{h}}[n-1]|^2+P|\mathbf{e}^*[n]\hat{\mathbf{h}}[n-1]|^2\notag\\
&\stackrel{(b)}{\leq}P{N_t}\rho^2\left(1-2^{-\frac{B}{N_t-1}}\right)+P\mathbb{E}\left[|\hat{\mathbf{h}}^*[n-1]\cdot[\mathbf{e}[n]\mathbf{e}^*[n]]\cdot\hat{\mathbf{h}}[n-1]|\right]\notag\\
&\stackrel{(c)}{=}P{N_t}\rho^2\left(1-2^{-\frac{B}{N_t-1}}\right)+P\epsilon_e^2,\notag
\end{align}
As $\mathbf{e}[n]$ is independent of $\mathbf{h}[n-1]$, it is also
independent of $\hat{\mathbf{h}}[n-1]$, which gives (a). Step (b)
follows \eqref{BF_Q}. Step (c) is from the fact
$\mathbf{e}[n]\sim\mathcal{CN}(\mathbf{0},\epsilon_e^2\mathbf{I}_{N_t})$
and $|\hat{\mathbf{h}}[n-1]|^2=1$.

\subsection{Proof of Theorem \ref{thmR_BFD}}\label{pthmR_BFD}
Denote $y_1=\|\mathbf{h}[n-1]\|^2$ and
$y_2=\frac{1}{\epsilon_e^2}|\mathbf{e}^*[n]\tilde{\mathbf{h}}[n-1]|^2$,
then $y_1\sim\chi^2_{2N_t}$, $y_2\sim\chi_2^2$, and they are
independent. The received SNR can be written as
$x=\eta_1{y_1}+\eta_2{y_2}$, where $\eta_1=P\rho^2$ and
$\eta_2=P\epsilon_e^2$. The cdf of $x$ is given as \cite{Sim02}
\begin{align}
F_X(x)&=1-\left(\frac{\eta_2}{\eta_2-\eta_1}\right)^{N_t}e^{-x/\eta_2}\notag\\
&\quad+e^{-x/\eta_1}\left(\frac{\eta_1}{\eta_2-\eta_1}\right)\cdot\sum_{i=0}^{N_t-1}\sum_{l=0}^i\frac{1}{(i-l)!}\left(\frac{\eta_2}{\eta_2-\eta_1}\right)^{N_t-1-i}\left(\frac{x}{\eta_1}\right)^{i-l}.
\end{align}
Denote $a_0=\frac{\eta_2}{\eta_2-\eta_1}$ and following \emph{Lemma
\ref{lemma_rate}} we have
\begin{align}
&\mathbb{E}_X\left[\ln(1+X)\right]\notag\\
=&\int_0^\infty\frac{1-F_X(x)}{1+x}dx\notag\\
=&a_0^{N_t}\int_0^\infty\frac{e^{x/\eta_2}}{1+x}dx-(1-a_0)\sum_{i=0}^{N_t-1}\sum_{l=0}^i\frac{a_0^{N_t-1-i}}{(i-l)!}\left(\frac{1}{\eta_1}\right)^{i-l}\int_0^\infty\frac{x^{i-l}e^{-x/\eta_1}}{1+x}dx\notag\\
=&a_0^{N_t}I_2(1/{\eta_2},1,1)-(1-a_0)\sum_{i=0}^{N_t-1}\sum_{l=0}^i\frac{a_0^{N_t-1-i}}{(i-l)!}\left(\frac{1}{\eta_1}\right)^{i-l}I_1(1/\eta_1,1,i-l).
\end{align}
where $I_1(\cdot,\cdot,\cdot)$ and $I_2(\cdot,\cdot,\cdot)$ are
given in \eqref{I1} and \eqref{I2}, respectively.

\subsection{Proof of Lemma
\ref{lemma_Iq}}\label{plemma_Iq} Let
$x=\|\mathbf{h}_u[n-D]\|^2\sin^2\theta\sim{\Gamma}(M-1,\delta)$,
$y\sim\beta(1,M-2)$, and $x$ is independent of $y$. Then the
interference term due to quantization is $z=xy$. The cdf of $z$ is
\begin{align}
P_Z(z)&=P(xy\leq{z})\notag\\
&=\int_0^\infty{F}_{Y|X}\left(\frac{z}{x}\right)f_X(x)dx\notag\\
&=\int_0^z{f}_X(x)dx+\int_z^\infty\left(1-\left(1-\frac{z}{x}\right)^{M-2}\right)f_X(x)dx\notag\\
&=\int_0^\infty{f}_X(x)dx-\int_z^\infty\left(1-\frac{z}{x}\right)^{M-2}x^{M-2}\frac{e^{-x/\delta}}{(M-2)!\delta^{M-1}}dx\notag\\
&=1-e^{-z/\delta}\int_z^\infty(x-z)^{M-2}\frac{e^{-(x-z)/\delta}}{(M-2)!\delta^{M-1}}dx\notag\\
&\stackrel{(a)}{=}1-e^{-z/\delta},
\end{align}
where step (a) follows the equality
$\int_0^\infty{y}^Me^{-\alpha{y}}=M!\alpha^{-(M+1)}$.

\subsection{Proof of Theorem \ref{thm_RQD}}\label{pthm_RQD}
Assuming each interference term in \eqref{ZFSINR_QD_approx} is
independent of each other and independent of the signal power term,
denote
$\sum_{u'\neq{u}}\rho_u^2|\mathbf{h}_u^*[n-1]\mathbf{f}^{(QD)}_{u'}[n]|^2=\rho_u^2\delta{y}_1$
and
$\sum_{u'\neq{u}}|\mathbf{e}_u^*[n]\mathbf{f}^{(QD)}_{u'}[n]|^2=\epsilon_{e,u}^2y_2$,
then from \emph{Lemma \ref{lemma_Iq}} we have
$y_1\sim\chi^2_{2(N_t-1)}$, and $y_2\sim\chi^2_{2(N_t-1)}$ as
$\mathbf{e}_u[n]$ is complex Gaussian with variance
$\epsilon_{e,u}^2$ and independent of the normalized vector
$\mathbf{f}^{(QD)}_{u'}[n]$. In addition, the signal power
$|\mathbf{h}^*_u[n]\mathbf{f}^{(QD)}_u[n]|^2\sim\chi^2_2$. Then the
received SINR for the $u$-th user is approximated as
\begin{equation}
\gamma_{ZF,u}^{(QD)}\approx\frac{\alpha{z}}{1+\beta(\delta_1y_1+\delta_2y_2)}\triangleq{x},
\end{equation}
where $\alpha=\beta=\frac{P}{U}$, $\delta_1=\rho_u^2\delta$,
$\delta_2=\epsilon_{e,u}^2$, $y_1\sim\chi^2_{2M}$,
$y_1\sim\chi^2_{2M}$, $M=N_t-1$, $z\sim\chi^2_2$, and $y_1$, $y_2$,
$z$ are independent of each other.

Let $y=\delta_1y_1+\delta_2y_2$, then the pdf of $y$, which is the
sum of two independent chi-square random variables, is given as
\cite{Sim02}
\begin{align}
p_Y(y)&=e^{-y/\delta_1}\sum_{i=0}^{M-1}a^{(1)}_iy^i+e^{-y/\delta_2}\sum_{i=0}^{M-1}a^{(2)}_iy^i\notag\\
&=\sum_{j=1}^2\sum_{i=0}^{M-1}e^{-y/\delta_j}a^{(j)}_iy^i,
\end{align}
where
\begin{align}
a^{(1)}_i&=\frac{1}{\delta_1^{i+1}(M-1)!}\left(\frac{\delta_1}{\delta_1-\delta_2}\right)^M\frac{(2(M-1)-i)!}{i!(M-1-i)!}\left(\frac{\delta_2}{\delta_2-\delta_1}\right)^{M-1-i}\label{eq_a1}\\
a^{(2)}_i&=\frac{1}{\delta_2^{i+1}(M-1)!}\left(\frac{\delta_2}{\delta_2-\delta_1}\right)^M\frac{(2(M-1)-i)!}{i!(M-1-i)!}\left(\frac{\delta_1}{\delta_1-\delta_2}\right)^{M-1-i}\label{eq_a2}.
\end{align}
The cdf of $x$ is
\begin{align}
F_X(x)&=P\left(\frac{\alpha{z}}{1+\beta{y}}\leq{x}\right)\notag\\
&=\int_0^\infty{F}_{Z|Y}\left(\frac{x}{\alpha}(1+\beta{y})\right)p_Y(y)dy\notag\\
&=\int_0^\infty\left(1-e^{-\frac{x}{\alpha}(1+\beta{y})}\right)p_Y(y)dy\notag\\
&=1-e^{-x/\alpha}\int_0^\infty{e}^{-\beta{xy}/\alpha}p_Y(y)dy\notag\\
&=1-e^{-x/\alpha}\int_0^\infty\left\{\sum_{j=1}^2\sum_{i=0}^{M-1}\exp\left[-\left(\frac{\beta}{\alpha}x+\frac{1}{\delta_j}\right)y\right]a^{(j)}_iy^i\right\}dy\notag\\
&\stackrel{(a)}{=}1-e^{-x/\alpha}\sum_{j=1}^2\sum_{i=0}^{M-1}\left[\frac{a^{(j)}_ii!}{\left(\frac{\beta}{\alpha}x+\frac{1}{\delta_j}\right)^{i+1}}\right],
\end{align}
where step (a) follows the equality
$\int_0^\infty{y}^Me^{-\alpha{y}}=M!\alpha^{-(M+1)}$.

Then the ergodic achievable rate for the $u$-th user is approximated
as
\begin{align}
R_{ZF,u}^{(QD)}&=\mathbb{E}_\gamma\left[\log_2\left(1+\gamma_{ZF,u}^{(QD)}\right)\right]\notag\\
&\approx\log_2(e)\mathbb{E}_X\left[\ln(1+X)\right]\notag\\
&\stackrel{(a)}{=}\log_2(e)\int_0^\infty\frac{1-F_X(x)}{x+1}dx\notag\\
&=\log_2(e)\int_0^\infty\sum_{i=0}^{M-1}\sum_{j=1}^2\left[a^{(j)}_ii!\left(\frac{\alpha}{\beta}\right)\frac{e^{-x/\alpha}}{\left(x+\frac{\alpha}{\beta\delta_j}\right)^{i+1}(x+1)}\right]\notag\\
&\stackrel{(b)}{=}\log_2(e)\sum_{i=0}^{M-1}\sum_{j=1}^2\left[a^{(j)}_ii!\left(\frac{\alpha}{\beta}\right)^{i+1}I_3\left(\frac{1}{\alpha},\frac{\alpha}{\beta\delta_j},i+1\right)\right],
\end{align}
where step (a) follows from \emph{Lemma \ref{lemma_rate}}, step (b)
follows the expression of $I_3(\cdot,\cdot,\cdot)$ in \eqref{I3}.

\bibliographystyle{IEEEtran}
\bibliography{CSIDelay}

\begin{table}[ht]
\caption{System Parameters}\label{Tbl1} \centering
\begin{tabular}{c|c}
\hline \hline Symbol & Description\\
\hline
\hline $N_t$ & number of transmit antennas\\
\hline $U$ & number of mobile users\\
\hline $B$ & number of feedback bits\\
\hline $L$ & quantization codebook size, $L=2^B$\\
\hline $P$ & average SNR\\
\hline $n$ & time index\\
\hline $T_s$ & the length of each symbol\\
\hline $f_d$ & the Doppler frequency\\
 \hline\hline
\end{tabular}
\end{table}

\begin{table}[ht]
\caption{Mode Switching Points}\label{MMSE} \centering
\begin{tabular}{c|c|c|c}
\hline \hline & $f_dT_s=0.03$ & $f_dT_s=0.04$ & $f_dT_s=0.05$\\
\hline
\hline MMSE (Simulation) & $44.2$ dB & $35.7$ dB & $29.5$ dB\\
\hline ZF (Simulation) & $44.2$ dB & $35.4$ dB & $28.6$ dB\\
\hline ZF (Calculation) & $41.6$ dB & $32.9$ dB & $26.1$ dB\\
 \hline\hline
\end{tabular}
\end{table}

\begin{figure}[htb]
\centering
\includegraphics[width=3.6in]{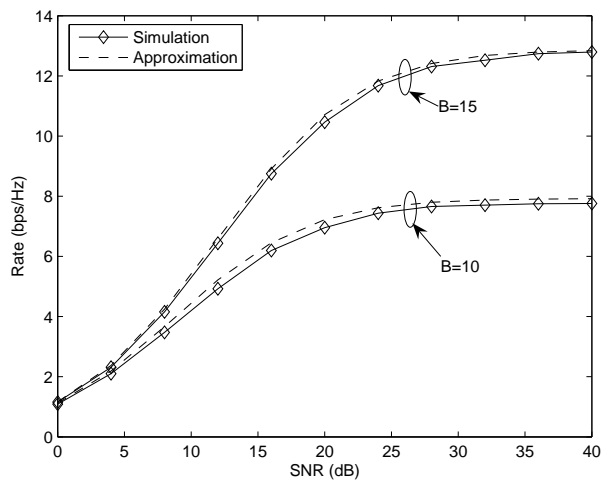}
\caption{Approximations and simulations for the ZF system with
limited feedback, $N_t=U=4$.}\label{fig_ZFLFB}
\end{figure}

\begin{figure}[htb]
\centering
\includegraphics[width=3.6in]{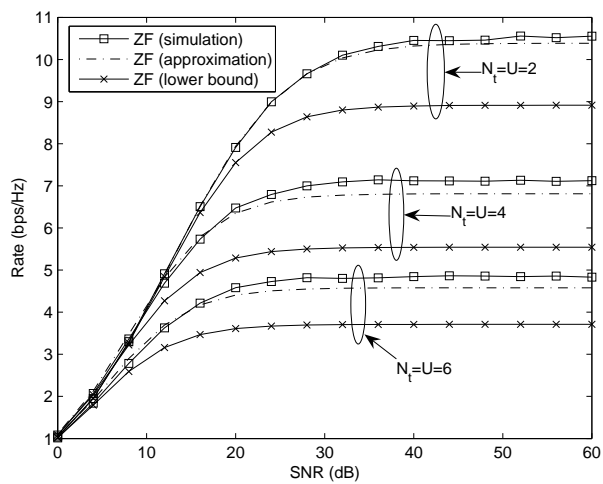}
\caption{Comparison of approximation in \eqref{R_QD}, the lower
bound in \eqref{R_ZFlb}, and the simulation results for the ZF
system with both delay and channel quantization. $B=10$, $f_c=2$
GHz, $v=20$ km/hr, and $T_s=1$ msec.}\label{fig_ZFSimvsCal}
\end{figure}

\begin{figure}[t]
\centering
\includegraphics[width=3.6in]{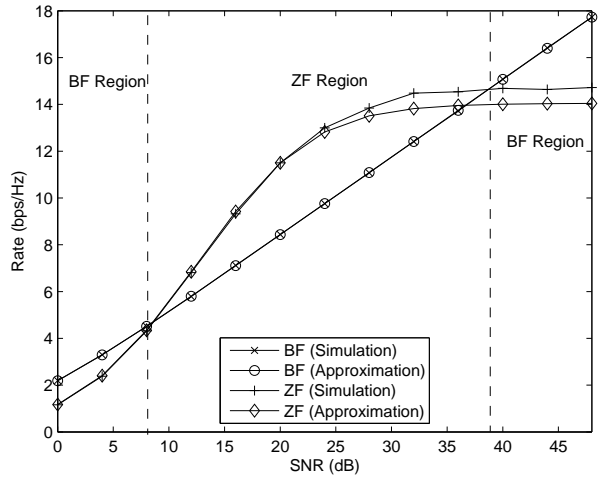}
\caption{Mode switching between BF and ZF modes with both CSI delay
and channel quantization, $B=18$, $N_t=4$, $f_c=2$ GHz, $T_s=1$
msec, $v=10$ km/hr.}\label{figBFZFswitch_QD}
\end{figure}

\begin{figure*}[htb]
\centerline{\subfigure[Different
$f_dT_s$.]{\includegraphics[width=3.4in]{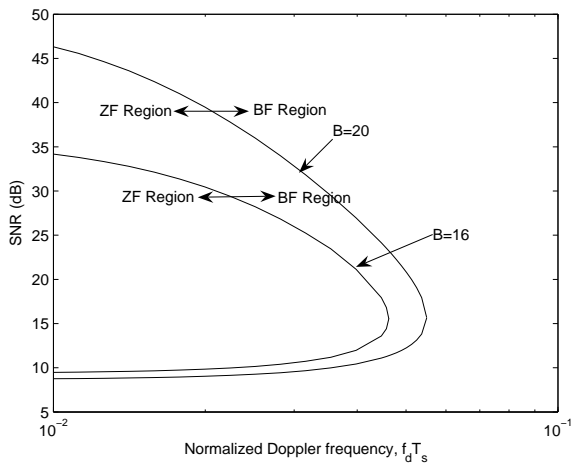}
\label{figBFZFswitchregion_DQ_SNRD}} \hfil \subfigure[Different $B$,
$f_c=2$ GHz, $T_s=1$
msec.]{\includegraphics[width=3.4in]{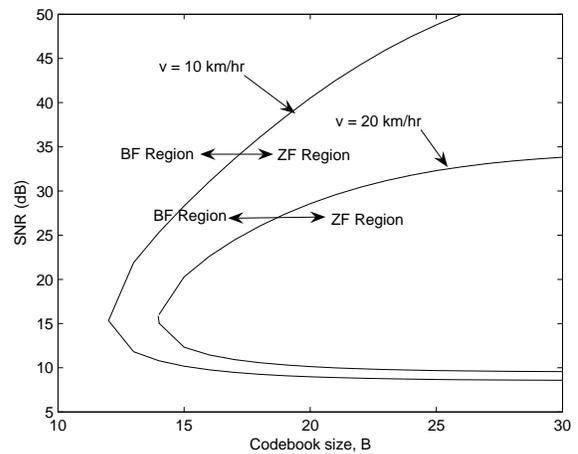}
\label{figBFZFswitchregion_DQ_SNRQ}}} \caption{Operating regions for
BF and ZF with both CSI delay and quantization, $N_t=4$.}
\label{figBFZFswitchregion_DQ}
\end{figure*}

\begin{figure}[htb]
\centering
\includegraphics[width=3.6in]{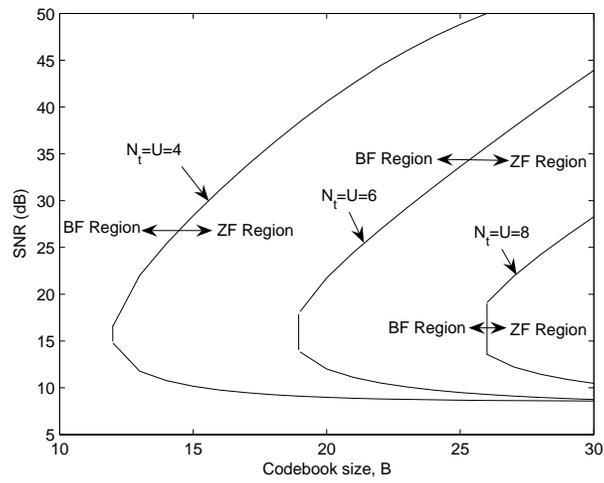}
\caption{Operating regions for BF and ZF with different $N_t$,
$f_c=2$ GHz, $v=10$ km/hr, $T_s=1$
msec.}\label{figBFZFswitchregion_DQ_U}
\end{figure}

\begin{figure}[htb]
\centering
\includegraphics[width=3.6in]{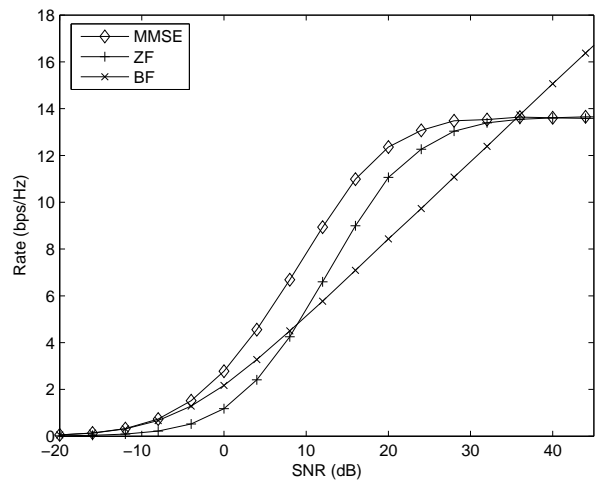}
\caption{Simulation results for BF, ZF and MMSE systems with delay,
$N_t=U=4$, $f_dT_s=0.04$.}\label{fig_MMSE}
\end{figure}

\end{document}